\documentclass{IEEEtran}

\usepackage{cite}
\usepackage{amsmath,amssymb,amsfonts,mathrsfs,amsthm}

\newtheoremstyle{compact}  
  {2pt}    
  {2pt}    
  {\normalfont}  
  {}       
  {\bfseries} 
  {.}      
  { }      
  {}       
\theoremstyle{compact}

\usepackage{algorithmic}
\usepackage{graphicx}
\usepackage{textcomp}

\def\BibTeX{{\rm B\kern-.05em{\sc i\kern-.025em b}\kern-.08em
    T\kern-.1667em\lower.7ex\hbox{E}\kern-.125emX}}

\usepackage{hyperref}
\hypersetup{hidelinks=true}
\usepackage{textcomp}
\usepackage{optidef}
\usepackage{booktabs}
\usepackage{multirow}
\usepackage{tikz,pgfplots}
\usepackage{color}
\pgfplotsset{compat=newest}
\definecolor{IEEEblue} {rgb}{0.0, 0.384, 0.608}
\definecolor{IEEEgray}{rgb}{0.459, 0.471, 0.482}
\definecolor{IEEEblack}{rgb}{0.137, 0.122, 0.125}
\definecolor{Darkgray}{rgb}{0.314, 0.314, 0.314}
\definecolor{brown}{rgb}{0.396, 0.263, 0.129}
\definecolor{green}{rgb}{0.333, 0.420, 0.184}
\newtheorem{thm}{Theorem}[section]
\newtheorem{lem}[thm]{Lemma}
\newtheorem{prop}[thm]{Proposition}
\newtheorem{rem}{Remark}
\newtheorem{assum}{Assumption}
\usepackage{orcidlink}  
\usepackage{dsfont}
\renewcommand{\thethm}{\arabic{section}.\arabic{thm}} 

\gdef\parenineq{true} 

\newcommand{\ifequals}[4]{\ifthenelse{\equal{#1}{#2}}{#3}{#4}}
\newcommand{\casex}[2]{#1 #2} 
\newenvironment{switch}[1]{\renewcommand{\casex}{\ifequals{#1} } } {}
\newcommand{\enmath}[1]{\ensuremath{#1}} 

\newcommand{\continuousLabel}[1]{ \bar{#1}}
\newcommand{\disturbedLabel}[1]{{#1}_\mrm{d}}

\newcommand{\optLabel}[1]{{#1}^{*}}
\newcommand{\warmLabel}[1]{\tilde{#1}}

\newcommand{\leftright}[4]{ 
\begin{switch}{#4}
        \casex{in}{\enmath{#1#3#2}}{
        \casex{inline}{\enmath{#1#3#2}}{
        \casex{off}{\enmath{#1#3#2}}{
        \casex{false}{\enmath{#1#3#2}}{
        \enmath{\mathopen{} \left#1 #3 \right#2 \mathclose{} }
        }}}}
        \end{switch}
 }

\NewDocumentCommand{\paren}{O{in} m}{ \leftright{(}{)}{#2}{#1}}
\NewDocumentCommand{\brackets}{O{in} m}{ \leftright{[}{]}{#2}{#1}}
\NewDocumentCommand{\cbraces}{O{in} m}{ \leftright{\{}{\}}{#2}{#1}}
\NewDocumentCommand{\abs}{O{in} m}{ \leftright{\vert}{\vert}{#2}{#1}}
\NewDocumentCommand{\norm}{O{in} m}{ \leftright{\Vert}{\Vert}{#2}{#1}}
\NewDocumentCommand{\maxnorm}{O{in} m}{ \leftright{\Vert}{\Vert_{\infty}}{#2}{#1}}

\newcommand{\mrm}[1]{ \enmath{\mathrm{#1}} }

\newcommand{\mbf}[1]{ \enmath{\mathbf{#1}} }

\newcommand{\mbb}[1]{ \enmath{\mathbb{#1}} }

\newcommand{\mcal}[1]{ \enmath{\mathcal{#1}} }

\NewDocumentCommand{\set}{O{in} m m}{ \enmath{ \cbraces[#1]{#2\,\vert\, #3} } }

\newcommand{\funcArgs}[2]{
\begin{switch}{#1}
        \casex{in}{\gdef\parenineq{in}\xfuncArgs{#2}}{
        \casex{inline}{\gdef\parenineq{inline}\xfuncArgs{#2}}{
        \casex{off}{\gdef\parenineq{off}\xfuncArgs{#2}}{
        \casex{false}{\gdef\parenineq{false}\xfuncArgs{#2}}{
        \casex{eq}{\gdef\parenineq{eq}\xfuncArgs{#2}}{
        \casex{on}{\gdef\parenineq{on}\xfuncArgs{#2}}{
        \casex{true}{\gdef\parenineq{true}\xfuncArgs{#2}}{
        {\xfuncArgs[#1]}
        }}}}}}}
\end{switch}
\gdef\parenineq{true} 
}  

\newcommand{\xfuncArgs}[1]{
\begin{switch}{#1}
        \casex{\empty}{}{
        \casex{ }{}{
        \casex{.}{(\cdot)}{
        {\paren[\parenineq]{#1}}
        }}}
\end{switch}
}  

\newcommand{\funcArgsNorm}[2]{
\begin{switch}{#1}
        \casex{in}{\gdef\parenineq{in}\xfuncArgsNorm{#2}}{
        \casex{inline}{\gdef\parenineq{inline}\xfuncArgsNorm{#2}}{
        \casex{off}{\gdef\parenineq{off}\xfuncArgsNorm{#2}}{
        \casex{false}{\gdef\parenineq{false}\xfuncArgsNorm{#2}}{
        \casex{eq}{\gdef\parenineq{eq}\xfuncArgsNorm{#2}}{
        \casex{on}{\gdef\parenineq{on}\xfuncArgsNorm{#2}}{
        \casex{true}{\gdef\parenineq{true}\xfuncArgsNorm{#2}}{
        {\xfuncArgsNorm[#1]}
        }}}}}}}
\end{switch}
\gdef\parenineq{true} 
}  

\newcommand{\xfuncArgsNorm}[1]{
\begin{switch}{#1}
        \casex{\empty}{\paren[inline]{\norm[inline]{\cdot}}}{
        \casex{ }{\paren[inline]{\norm[inline]{\cdot}}}{
        \casex{.}{\paren[inline]{\norm[inline]{\cdot}}}{
        {\paren[\parenineq]{ \norm[\parenineq]{#1} }  }
        }}}
\end{switch}
}  

\NewDocumentCommand{\JN}{O{in} m}{ \enmath{ J_{N} \funcArgs{#1}{#2} } }
\NewDocumentCommand{\Jfc}{O{in} m}{ \enmath{ J_\mrm{f} \funcArgs{#1}{#2} } }
\NewDocumentCommand{\xell}{O{in} m}{ \enmath{ \ell \funcArgs{#1}{#2} } }
\NewDocumentCommand{\xellc}{O{in} m}{ \enmath{{\ell}_{x} \funcArgs{#1}{#2} } }
\NewDocumentCommand{\vellc}{O{in} m}{ \enmath{{\ell}_{v} \funcArgs{#1}{#2} } }
\NewDocumentCommand{\uellc}{O{in} m}{ \enmath{{\ell}_{u} \funcArgs{#1}{#2} } }
\NewDocumentCommand{\xsolv}{O{in} m}{ \enmath{ \varphi_{\deltat}^{v} \funcArgs{#1}{#2} } }
\NewDocumentCommand{\xsol}{O{in} m}{ \enmath{ \varphi_{\deltat} \funcArgs{#1}{#2} } }
\NewDocumentCommand{\xsolos}{O{in} m}{ \enmath{ \varphi^{\deltatos}_{N_{\os}} \funcArgs{#1}{#2} } }
\NewDocumentCommand{\xsolcl}{O{in} m}{ \enmath{ {\phi}_{\deltat} \funcArgs{#1}{#2} } }
\NewDocumentCommand{\xsolcld}{O{in} m}{ \enmath{ \check{\phi}_{\deltat} \funcArgs{#1}{#2} } }
\NewDocumentCommand{\VN}{O{in} m}{ \enmath{ V_{N} \funcArgs{#1}{#2} } }
\newcommand{\levJf}[1]{ \enmath{ \mrm{lev}_{#1} \Jfc{} } }
\newcommand{\levVN}[1]{ \enmath{ \mrm{lev}_{#1} \VN{} } }
\NewDocumentCommand{\xc}{O{in} m}{\enmath{ \continuousLabel{x} \funcArgs{#1}{#2} }}
\NewDocumentCommand{\vc}{O{in} m}{ \enmath{ \continuousLabel{v} \funcArgs{#1}{#2} } }
\NewDocumentCommand{\f}{O{in} m}{ \enmath{ f_{\deltat} \funcArgs{#1}{#2} } }
\NewDocumentCommand{\fos}{O{in} m}{ \enmath{ f_{\deltatos} \funcArgs{#1}{#2} } }
\NewDocumentCommand{\xmu}{O{in} m}{ \enmath{ \mu \funcArgs{#1}{#2} } }
\NewDocumentCommand{\xmuf}{O{in} m}{ \enmath{ \mu_\mrm{f} \funcArgs{#1}{#2} } }
\NewDocumentCommand{\xmud}{O{in} m}{ \enmath{ \disturbedLabel{\mu} \funcArgs{#1}{#2} } }
\NewDocumentCommand{\F}{O{in} m}{ \enmath{ F_{\deltat} \funcArgs{#1}{#2} } }
\NewDocumentCommand{\Fos}{O{in} m}{ \enmath{ F_{\deltatos} \funcArgs{#1}{#2} } }
\NewDocumentCommand{\Fc}{O{in} m}{ \enmath{ \continuousLabel{F} \funcArgs{#1}{#2} } }
\NewDocumentCommand{\fc}{O{in} m}{ \enmath{ f \funcArgs{#1}{#2} } }
\newcommand{\deltat}{ \enmath{ \Delta t} }
\newcommand{\deltatos}{ \enmath{ \delta t } }

\NewDocumentCommand{\alphaone}{O{in} m}{ \enmath{ \alpha_{1} \funcArgs{#1}{#2} } }
\NewDocumentCommand{\alphaonenorm}{O{in} m}{ \enmath{ \alpha_{1} \funcArgsNorm{#1}{#2} } }
\NewDocumentCommand{\alphatwo}{O{in} m}{ \enmath{ \alpha_{2} \funcArgs{#1}{#2} } }
\NewDocumentCommand{\alphatwonorm}{O{in} m}{ \enmath{ \alpha_{2} \funcArgsNorm{#1}{#2} } }
\NewDocumentCommand{\alphathree}{O{in} m}{ \enmath{ \alpha_{3} \funcArgs{#1}{#2} } }
\NewDocumentCommand{\alphathreenorm}{O{in} m}{ \enmath{ \alpha_{3} \funcArgsNorm{#1}{#2} } }
\NewDocumentCommand{\alphaell}{O{in} m}{ \enmath{ \alpha_{\xell{}} \funcArgs{#1}{#2} } }
\NewDocumentCommand{\alphaellnorm}{O{in} m}{ \enmath{ \alpha_{\xell{}} \funcArgsNorm{#1}{#2} } }
\NewDocumentCommand{\alphaJf}{O{in} m}{ \enmath{ \alpha_{\Jf} \funcArgs{#1}{#2} } }

\NewDocumentCommand{\alphaoneinv}{O{in} m}{ \enmath{ \alpha^{-1}_{1} \funcArgs{#1}{#2} } }
\NewDocumentCommand{\alphatwoinv}{O{in} m}{ \enmath{ \alpha^{-1}_{2} \funcArgs{#1}{#2} } }

\newcommand{\x}[1]{ \enmath{ x_{#1} } }
\newcommand{\xo}{ \enmath{ x_{0} } }
\newcommand{\xf}{ \enmath{ x_\mrm{f} } }
\newcommand{\omegaf}{ \enmath{ \omega_\mrm{f} } }
\newcommand{\vf}{ \enmath{ v_\mrm{f} } }
\newcommand{\xplus}{ \enmath{ x_{+} } }
\newcommand{\xd}[1]{ \enmath{\check{x}_{#1}}}
\newcommand{\xdplus}{ \enmath{\check{x}_{+}} }
\renewcommand{\u}[1]{ \enmath{ u_{#1} } }
\newcommand{\uopt}[1]{ \enmath{ \optLabel{u}_{#1} } }
\newcommand{\us}{ \enmath{ \mathbf{u} } }
\newcommand{\omegas}{ \enmath{ \boldsymbol{\omega} } }
\NewDocumentCommand{\omegasmu}{O{in} m}{ \enmath{ \omegas_{\mu} \funcArgs{#1}{#2} } }
\newcommand{\vs}{ \enmath{ \mathbf{v} } }
\NewDocumentCommand{\usopt}{O{in} m}{ \enmath{ \optLabel{\mathbf{u}} \funcArgs{#1}{#2} } }
\NewDocumentCommand{\uswarm}{O{in} m}{ \enmath{ \warmLabel{\mathbf{u}} \funcArgs{#1}{#2} } }
\newcommand{\uf}{ \enmath{ u_\mrm{f} } }
\newcommand{\ufi}{ \enmath{ {u_{\mrm{f}}}_{i} } } 

\newcommand{\mbbN}{\enmath{\mbb{N}}}
\newcommand{\mbbR}{\enmath{\mbb{R}}}
\newcommand{\mbbS}{\enmath{\mbb{S}}}
\newcommand{\mbbU}{\enmath{\mbb{U}}}
\newcommand{\mbbV}{\enmath{\mbb{V}}}
\newcommand{\mbbX}{\enmath{\mbb{X}}}
\newcommand{\terminalset}{\enmath{\mbbX_\mrm{f} }}
\newcommand{\mcalA}{\enmath{\mcal{A}}}
\newcommand{\mcalB}{\enmath{\mcal{B}}}
\newcommand{\mcalC}{\enmath{\mcal{C}}}
\newcommand{\mcalK}{\enmath{\mcal{K}}}
\newcommand{\mcalL}{\enmath{\mcal{L}}}
\newcommand{\mcalP}{\enmath{\mcal{P}}}
\newcommand{\mcalV}{\enmath{\mcal{V}}}

\newcommand{\feasibleset}{\enmath{\mathcal{X}_{N}}}
\newcommand{\feasiblesetc}{\enmath{\mathcal{X}_{\tf}}}
\NewDocumentCommand{\admissiblecontrolset}{O{in} m}{\enmath{{\mathcal{U}}_{N} \funcArgs{#1}{#2}}}

\DeclareMathOperator*{\argmax}{arg\,max}

\newcommand{\dt}{\enmath{\,\mrm{d}t } }
\newcommand{\dtau}{\enmath{\,\mrm{d}\tau } }
\newcommand{\tf}{\enmath{t_\mrm{f} } }
\newcommand{\os}{\enmath{\mrm{os}}}
\newcommand{\nv}{\enmath{n_{v}}}
\newcommand{\nx}{\enmath{n_{x}}}  

\newcommand{\leftlabel}[1]{%
  \makebox[0pt][l]{\hspace*{-\dimexpr\oddsidemargin+1in}#1}%
}

\begin{document}
\title{Fast Switching in Mixed-Integer Model Predictive Control}
\author{Artemi Makarow \orcidlink{0000-0002-0822-1807}, \IEEEmembership{Member, IEEE}, Christian Kirches \orcidlink{0000-0002-3441-8822}
\thanks{Funded by the European Union. The views and opinions expressed are those of the author(s) only and do not necessarily reflect those of the European Union or the European Research Council Executive Agency. Neither the European Union nor the granting authority can be held responsible for them. This work is supported by ERC grant SCARCE,
101087662.}
\thanks{The authors are with the Institute for Mathematical Optimization, Technische Universit\"at Braunschweig, 38106 Braunschweig, Germany, {\tt\small \{artemi.makarow,c.kirches\}@tu-bs.de}. }}

\maketitle
\thispagestyle{empty}
\pagestyle{empty}


\begin{abstract}                          

We deduce stability results for finite control set and mixed-integer model predictive control with a downstream oversampling phase.
The presentation rests upon the inherent robustness of model predictive control with stabilizing terminal conditions and techniques for solving mixed-integer optimal control problems by continuous optimization. 
Partial outer convexification and binary relaxation transform mixed-integer problems into common optimal control problems. 
We deduce nominal asymptotic stability for the resulting relaxed system formulation and implement sum-up rounding to restore efficiently integer feasibility on an oversampling time grid. 
If fast control switching is technically possible and inexpensive, we can approximate the relaxed system behavior in the state space arbitrarily close. 
We integrate input perturbed model predictive control with practical asymptotic stability. 
Numerical experiments illustrate practical relevance of fast control switching.

\end{abstract}

\begin{IEEEkeywords}
Model Predictive Control, Integer Approximation, Sum-Up Rounding, Practical Asymptotic Stability.
\end{IEEEkeywords}



\section{Introduction}

Formulating mixed-integer optimal control problems (MI-OCPs) for nonlinear dynamical systems significantly increases complexity compared to conventional OCPs. 
For practical applications, we can either transcribe a MI-OCP into a mixed-integer nonlinear program (MI-NLP) and apply the computational expensive branch-and-bound method as in \cite{Gerdts2005} or we follow the numerically challenging variable time transformation in \cite{Gerdts2006} that is based on a pre-defined switching sequence, see also \cite{Sager2007}. 
For a survey on reformulating and solving generic MI-OCPs, refer to \cite{Sager2009}. 

In this paper, we mainly focus on the integer approximation framework originally presented in \cite{Sager2007,Sager2010}. 
This framework applies three steps to generate an integer feasible but suboptimal solution to a MI-OCP: Partial outer convexification, relaxation, and integer reconstruction via sum-up rounding (SUR). 
Hence, we first transcribe a MI-OCP into a common OCP and then round its solution back into an integer feasible control trajectory in polynomial time. 
Sager et al.  \cite{Sager2010} show that the input and state approximation error are upper bounded and depend linearly on the largest sampling width. 
These dependencies on the maximum sampling width motivate fast switching. 
The tightest error bound for SUR follows from a dynamic programming argument \cite{Kirches2020}. 

Nonlinear mixed-integer model predictive control (MI-MPC) further lifts the complexity, as the question of stabilization now also arises. 
It is easy to imagine that stabilization of a steady-state for a nonlinear system with discrete actuators is a challenging task. 
However, Rawlings and Risbeck \cite{Rawlings2017} state by their Folk Theorem that the stability results for conventional MPC with stabilizing terminal conditions also hold for systems with continuous- and discrete-valued inputs. 
The key reason for this statement is that the input constraint set does not need to have an interior and thus permits some integer controls. 
\mbox{MI-MPC} continues to be a relevant topic, and the state-of-the-art is analyzed and discussed in \cite{McAllister2022}. 

Partial outer convexification can be used to transform every finite control set OCP (FCS-OCP) into a binary-integer OCP \cite{Sager2010}. 
In this case, MPC can only resort to integer (discrete-valued) controls for solving the stabilization task. 
The authors in \cite{Picasso2003,Aguilera2011,Aguilera2013} derive stabilizing properties of FCS-MPC for linear time-invariant systems based on robust control analysis and the construction of invariant sets. 
FCS-MPC is further employed in power electronics and is usually limited to a one-step horizon to satisfy real-time constraints with combinatorial optimization, see, e.g., \cite{Karamanakos2020}. 

We, on the other hand, aim to integrate the integer approximation framework due to \cite{Sager2010} with the inherently robust MPC as derived in \cite{Yu2014,Allan2017}. 
Our idea is to design nominal MPC in the relaxed domain and then determine maximum input perturbation bounds for robust control. 
The authors in \cite{Ebrahim2024} follow a similar idea, however, they derive stochastic tube-based MPC, where an additive disturbance term models the uncertainty induced by the SUR. 
An additional tracking controller is designed to robustly track the nominal and relaxed reference trajectories. 
Related to our idea is also the work in \cite{Ebrahim2020}. 
Here, the authors propose a computationally demanding bi-level approach for switching systems. A second auxiliary and variable time OCP is used to minimize the impact of the rounding decisions induced by the solution of the switching minimizing mixed-integer linear program due to \cite{Sager2011} on some first part of the relaxed and optimal control trajectory. 
Recursive feasibility only follows implicitly from a weak assumption that the terminal state constraint is always satisfied for every small input perturbation on the first part. 

Implementing integer approximation approaches in the context of MI-MPC is not a new idea, however, has not yet been addressed in the context of robust asymptotic stability. 
Recent publications on MI-MPC, relying on integer approximation strategies according to \cite{Sager2010,Sager2011}, solve real-world problems as, e.g., smart building heating \cite{Burda2023} or controlling refrigeration systems \cite{Ebrahim2018}. 
Other existing contributions in this field of research investigate efficient numerical realizations as in \cite{Kirches2011,Buerger2019,Chen2022}. 
The work in \cite{Kirches2011} addresses, inter alia, closed-loop stability for shrinking horizon MI-MPC that is based on the real-time-iteration scheme due to \cite{Diehl2005}.

\textit{Contributions}. In Section \ref{sec:problem_formulation}, we present our problem formulation and the transformation process from a finite control set, as it arises in   
FCS-OCPs or MI-OCPs, to continuous valued controls via partial outer convexification and relaxation, leading to a discrete-time OCP in the relaxed domain of the convex multiplier. 
At the end of Section~\ref{sec:problem_formulation}, we introduce a downstream temporal oversampling grid to reduce the state approximation error in the integer reconstruction phase.  
In Section \ref{sec:set_point_stabilization}, we exploit the inherent robustness properties of conventional MPC due to \cite{Allan2017} and adapt them to the case of input perturbations resulting from control rounding. 
As a result, we derive \mbox{$\mcalP$-practical} asymptotic stability according to \cite{Gruene2017}. 
Section \ref{sec:input_rounding} introduces SUR on the oversampling grid and shows that there always exists a temporal resolution for which we can ensure $\mcalP$-practical asymptotic stability. 
In Section \ref{sec:simulation}, we substantiate our claims with a numerical example and show practical relevance of fast switching. 

\textit{Notation}. Let $\mbbR^{n}$ denote the $n$-dimensional real vector space equipped with the Euclidean norm $\norm{\cdot}$. 
The set of positive real numbers containing zero is denoted by $\mbbR_{0}^{+}$. 
The set of all natural numbers $\mbbN$ including zero is denoted by $\mbbN_{0}$.
$\abs{\Omega}$ represents the cardinality of a finite set $\Omega$.
Let $[N]\coloneqq\cbraces{0,1,\hdots,N-1}$ with $N\in\mbbN$.
Let $X$, $U$, and $Y$ be some subsets of the Euclidean space.
$\mcalC^{k}(X,Y)$ denotes the space of $k$-times differentiable functions $f:X \to Y$. 
The level set of a function $f:X \to \mbbR_{0}^{+}$ with some finite level $c\geq0$ is denoted by $\mathrm{lev}_{c}f\coloneqq \set{x \in X}{f(x)\leq c}$.
Let $\mcalB_{\delta} \coloneqq \set{x\in X}{\norm{x}\leq \delta}$ be the closed ball of radius $\delta>0$. 
If a property holds everywhere except on a set of Lebesgue measure zero, it holds almost everywhere (a.e.).
The space of absolute continuous functions on the time interval $[0,\tf]$ mapping to $Y$, which are differentiable a.e., is denoted by $\mcalA \mcalC\paren{\brackets{0 , \tf},Y}$.
We denote the space of piecewise constant (control) functions on a uniform time grid by $\mcalP\mcalC\paren{\brackets{0 , \tf},U}:= \set{f:\brackets{0, \tf}\to U}{\exists\, \cbraces{u_{k}}_{k=0}^{N-1} \subset U: f(t)=u_{k}, \forall \,t \in [k\deltat,(k+1)\deltat), \deltat\coloneqq\tf/N}$. 
For $R\in{\mbbR_{0}^{+}}^{1 \times n}$, we define $\mathrm{ker} (R) \coloneqq\set{x\in[0,1]^{n}}{\sum_{i=1}^{n}x_{i}=1\wedge R\,x=0}$.
Let us further define the following classes of comparison functions: $\mcalK := \set{\alpha \in \mcalC^{0}(\mbbR_{0}^{+},\mbbR_{0}^{+}) }{\forall \,x_{1},x_{2} \in \mbbR_{0}^{+}\,(x_{1} < x_{2} \implies \alpha(x_{1})<\alpha(x_{2}) ), \alpha(0)=0  }$, $\mcalK_{\infty} \coloneqq \set{\alpha \in \mcalK}{\lim_{x\to \infty} \alpha(x) = \infty}$, $\mcalL:=\set{\lambda \in \mcalC^{0}(\mbbR_{0}^{+},\mbbR_{0}^{+})}{\forall \,x_{1},x_{2} \in \mbbR_{0}^{+}\, (x_{1} < x_{2} \implies \lambda(x_{1})>\lambda(x_{2}) ),\lim_{x\to\infty}\lambda(x)=0}$, $\mcalK \mcalL:=\set{\beta \in \mcalC^{0}(\mbbR_{0}^{+}\times \mbbR_{0}^{+},\mbbR_{0}^{+})}{\beta(\cdot,y)\in\mcalK,\beta(x,\cdot)\in \mcalL}$. 
The $j$th unit vector is defined by $\mbf{1}^{j}$. 
The indicator function $\mathds{1}_{[a,b)}(x) = 1$ if $x\in[a,b)$, otherwise it is equal to zero.


\section{Problem Formulation}
\label{sec:problem_formulation}

The following nonlinear ordinary differential equation defines the dynamical system of interest with state $\xc{t}\in X\coloneqq \mbbR^{\nx}$, control $\vc{t}\in V\coloneqq\mbbR^{\nv}$, and time $t\in\mbbR_{0}^{+}$:
\begin{equation}
\frac{\mathrm d \xc{}}{\dt}(t) =\fc[eq]{\xc{t},\vc{t}} , ~\xc{0}= x.
\label{eq:ivp}
\end{equation}
We consider input constraints indicated by the set $\mbbV\subset V$.  
Let $\vc{} \in \mcalV := \mcalP\mcalC\paren{\brackets{ 0, \tf},\mbbV}$ be a piecewise constant control function that is defined by the control sequence \mbox{$\cbraces{v_{k}}_{k=0}^{N-1}\subset \mbbV$}. 
The state trajectory $\xc{} \in \mcalA\mcalC\paren{\brackets{ 0, \tf},X}$ is governed by the continuous vector field $\fc{} \in\mcalC^{0}\paren{X \times V,X}$, the initial state $x \in X$, and follows from solving the initial value problem in~\eqref{eq:ivp}.
We request that the last state at time~$\tf$ is an element of some terminal set $\terminalset\subseteq X$. 
The set of all feasible initial states is therefore defined by: 
\begin{equation}
\feasiblesetc \coloneqq \set[eq]{x }{\exists \,\vc{} \in\mcalV : x+{\int_{0}^{\tf}} \fc{ { \xc{t}, \vc{t} } } \dt \in \terminalset}.
\end{equation}

\subsection{Time Discretization}
Let us define an equivalent discrete-time system for all $k\in[N]$ with $\deltat\coloneqq \tf/N$ and $\x{k}\coloneqq \xc{k\deltat}$ as follows:
\begin{equation}
\x{k+1}=\f{ { \x{k} , v_{k} } } \coloneqq \x{k} +\int_{k\deltat}^{(k+1)\deltat} \fc{ { \xc{t}, v_{k} } } \dt.
\label{eq:sys_disc}
\end{equation}
The recursion $\xsolv{k,x,\mbf{v}} \coloneqq \f{ {\xsolv{k-1,x,\vs},v_{k-1}} }$ for all $k\in\cbraces{1,2,\dots,N}$ with $\xsolv{0,x,\mbf{v}} \coloneqq x$ describes the evolution of the discrete-time system over a finite horizon $N$ using the stacked control sequence \mbox{$\vs\coloneqq[v_{0},v_{1},\hdots,v_{N-1}]\in\mbbV^{N}$}.

Let us consider continuous stage cost functions $\xellc{} \in \mcalC^{0}(X,\mbbR_{0}^{+})$, $\vellc{} \in \mcalC^{0}(V,\mbbR_{0}^{+})$ over the prediction horizon~$N$, and a terminal cost function $\Jfc{}\in \mcalC^{0}(X,\mbbR_{0}^{+})$:
\begin{equation}
I_{N}({x,\vs}) \coloneqq \sum_{k=0}^{N-1}l(\xsolv{k,x,\vs},v_{k})+\Jfc[eq]{ { \xsolv{N,x,\vs} } }, 
\label{eq:orig_cost}
\end{equation}
where $l(x_{k},v_{k})\coloneqq\xellc{x_{k}} +\vellc{v_{k}}$.
The following two common design Assumptions \ref{assum:steady-state} and \ref{assum:stage_cost_bounds} are essential for the upcoming set-point stabilization task \cite{Mayne2000,Rawlings2020}.

\begin{assum}[Steady-State Behavior]\label{assum:steady-state} 
For some steady-state $\paren{\xf,\vf}\in \feasiblesetc\times \mbbV$, we have that $\fc{\xf,\vf}=0$ (thus $\f{ { \xf , \vf } } =\vf$), $\xellc{\xf}=0$, $\vellc{\vf}=0$, and $\Jfc{\xf}=0$. 
\end{assum}
Without loss of generality, we set the steady-state tuple $\paren{\xf,\vf}$ to $\paren{0,0}$.

\begin{assum}[Stage Cost Bounds]\label{assum:stage_cost_bounds} 
There exist functions $\alpha_{\xellc{}},\alpha_{\vellc{}} \in\mcalK_{\infty}$ such that for all $x\in \feasiblesetc$ and $v \in \mbbV$, we have that $\alpha_{\xellc{}}(\norm{x}) \leq \xellc{x}$ and $\alpha_{\vellc{}}(\norm{v}) \leq \vellc{v}$.
\end{assum}

\subsection{Finite Control Set and Partial Outer Convexification}
We want to handle finite control sets of the form  $\Omega\coloneqq\cbraces{v^1,v^2,..., v^{\abs{\Omega} }} \subset\mbbV $ with cardinality $2\leq \abs{\Omega}<\infty$, where it is not relevant whether the elements of the set $\Omega$ themselves are continuous- or discrete-valued controls. Though we have a finite control set $\Omega$, we strive for
continuous optimization and therefore rely on the integral approximation framework presented in \cite{Sager2007,Sager2010}.
Integer approximation due to \cite{Sager2007,Sager2010, Kirches2020} provides efficient integer reconstruction in polynomial time with tightly bounded approximation errors.

The compact set of convex multipliers satisfying the special ordered set of type 1 (SOS1) is given by \cite{Sager2010,Kirches2020}: 
\begin{equation}
\mbbS^{\abs{\Omega}}:=  \set[in]{s\in \cbraces{0,1}^{\abs{\Omega}}}{\textstyle\sum_{i=1}^{\abs{\Omega}} s_{i} = 1 }.
\label{eq:mult_vectors}
\end{equation}
The SOS1 constraint establishes a bijection between the finite control set $\Omega$ and the convex multiplier set $\mbbS^{\abs{\Omega}}$.
We reformulate the system dynamics $\f{}$ by successively substituting all elements of the finite control set $\Omega$
and implementing convex multipliers $\omega_{k} \in \mbbS^{\abs{\Omega}}$  for all $k\in [N]$ with the continuous mapping $\F{}: X \to \mbbR^{\nx\times\abs{\Omega}}$ \cite{Sager2007}:
\begin{equation}
\F{ x_{k} }\,\omega_{k}\coloneqq \sum_{i=1}^{\abs{\Omega}} \f[eq]{ { x_{k}, v^{i} } }\,\omega_{k,i}=\f{x_{k},v_{k} }.
\label{eq:sys_disc_convex}
\end{equation}
Notice that the convex multipliers $\omega_{k}$ are binary-valued. To apply continuous optimization, we therefore apply convex hull relaxation and obtain the following compact set \cite{Sager2010,Kirches2020}:
\begin{equation}
\mbbU^{\abs{\Omega}}:=  \set[in]{u\in \brackets{0,1}^{\abs{\Omega}}}{\textstyle\sum_{i=1}^{\abs{\Omega}} u_{i} = 1 }.
\label{eq:relaxed_mult}
\end{equation}
The following Assumption \ref{assum:constraint_sets} is important to ensure the existence of an admissible solution and also addresses the stabilization of integer infeasible steady-states. 
\begin{assum}[Constraint Sets]\label{assum:constraint_sets}
The terminal set  $\terminalset \coloneqq \levJf{\pi}$ with some $\pi>0$ contains $\xf$ in its interior. 
The compact set $\mbbU^{\abs{\Omega}}$ contains $\uf$ for which that the steady-state condition $\F{ { \xf} }\,\uf = \xf$ holds.
\end{assum}
We also apply outer convexification and relaxation to the cost function $\vellc{}$ with $R \in  (\mbbR_{0}^{+})^{1\times\abs{\Omega}}$ and $\uellc{}:\mbbU^{\abs{\Omega}} \to \mbbR_{0}^{+}$:
\begin{equation}
\uellc{ u_{k} }\coloneqq R\,u_{k}\coloneqq\sum_{i=1}^{\abs{\Omega}}\vellc{ {v^{i} } } \,u_{k,i}.
\label{eq:linear_cost}
\end{equation}
Let $\us \coloneqq \brackets{\u{0},\u{1},\ldots,\u{N-1}}\in (\mbbU^{\abs{\Omega}})^{N}$ denote a stacked sequence of relaxed convex multipliers, which we shall consider as the new control variables.
The recursive solution to the reformulated system dynamics is defined by $\xsol{k,x,\us} \coloneqq  \F{\xsol[eq]{k-1,x,\us}}\,u_{k-1}$ for all $k\in\cbraces{1,2,\ldots,N}$ with $\xsol{0,x,\us} \coloneqq x$.
Let us define the set of all admissible control sequences and the resulting feasible state space by:
\begin{align}
\admissiblecontrolset[in]{x}&\coloneqq \set[eq]{\us \in \paren{\mbbU^{\abs{\Omega}}}^{N} }{\xsol[in]{N,x,\us}\in\terminalset},\\
\feasibleset &\coloneqq \set{x\in X}{\admissiblecontrolset{x} \neq \emptyset} \subset \feasiblesetc.
\end{align}

The overall finite horizon cost function is defined by:
\begin{equation}
\JN{x,\us} \coloneqq \sum_{k=0}^{N-1}\xell[eq]{ \xsol{k,x,\us},u_{k} }+\Jfc[eq]{ { \xsol{N,x,\us}}},
\end{equation}
with $\xell{x_{k},u_{k}}\coloneqq \xellc{x_{k}} +\uellc{u_{k}}$.
Now, consider the following discrete-time OCP:
\begin{equation}
\leftlabel{(DT-OCP)}\VN[in]{x} \coloneqq \min_{\us\,\in\,\admissiblecontrolset[in]{x}} \JN[in]{x,\us}.
\label{eq:dt-ocp}
\end{equation}
We assume that DT-OCP \eqref{eq:dt-ocp} is well-defined since the mappings $\us \mapsto \xsol{k,x,\us}$ and $\us \mapsto \JN{x,\us} $ are continuous and valid on the compact set $\admissiblecontrolset{x} \neq \emptyset$ (see, e.g., \cite[Prop. 2.4]{Rawlings2020}). 
We have that $\VN{}: \feasibleset \to \mbbR_{0^{+}}$.
The optimal solution to DT-OCP \eqref{eq:dt-ocp} is denoted by $\usopt{x} \coloneqq \brackets[in]{\uopt{0}(x),\uopt{1}(x),\ldots,\uopt{N-1}(x)}\in \admissiblecontrolset{x}$. 

\subsection{Null Space of $R$}
If $\vf \in \Omega$, then by bijection there is a $\omegaf \in \mbbS^{\abs{\Omega}}$ satisfying $\F{ {\xf } }\,\omegaf = \f{ { \xf, \vf } } = \xf$. By Assumptions~\ref{assum:steady-state} and~\ref{assum:stage_cost_bounds}, we also have that $\uellc{ {\omegaf } }=R\,\omegaf =\vellc{ {\vf } }=0$. 
Hence, $\omegaf =\mathbf{1}^j$ is the only SOS1 admissible sample in the null space of $R$, which is $\ker(R) = \cbraces{0,\epsilon \,\mathbf{1}^j}$ with $\epsilon \geq 0$  if $j$ indexes $\vf$. 

The relaxation step invalidates the previously described bijective mapping, such that a relaxed steady-state control vector $\uf \in \mbbU^{\abs{\Omega}}\setminus \mbbS^{\abs{\Omega}}$ has no inverse control vector $\vf \in \Omega$ anymore. 
In this case, we observe that \mbox{$\uellc{ {\uf } }=R\,\uf>0$}. However, since $\F{ {\xf } }\,\uf = \xf$ with $\uf \notin \mbbS^{\abs{\Omega}}$ might be feasible, e.g., for input affine systems, we shall perform a coordinate transformation to introduce a trivial steady-state. 
Let us redefine the stage cost function $\uellc{}$ with $ \xi_{k}\coloneqq u_{k}-\uf$ for all $k \in [N]$ as:
\begin{equation}
\uellc{u_{k}} \coloneqq R \, \abs[eq]{\xi_{k}} \coloneqq\sum_{i=1}^{\abs{\Omega}}\vellc{ {v^{i} } }\,\abs{u_{k,i}-\ufi}.
\label{eq:abs_mult}
\end{equation}
Now, we have that $\uellc{ {\uf } }=0$ applies for all $\uf \in \mbbU^{\abs{\Omega}}$.
Notice that $\sum_{i=1}^{\abs{\Omega}}\vellc{ {v^{i} } }\,\abs{u_{k,i}-\ufi}= \sum_{i=1}^{\abs{\Omega}}\vellc{ {v^{i} } }\,u_{k,i}$ if $\uf \in\mbbS^{\abs{\Omega}}$ since $R_{j}=0$. If $\uf \not\in\mbbS^{\abs{\Omega}}$, \eqref{eq:abs_mult} introduces a cost regularization term.

\subsection{Downstream Temporal Oversampling}
\label{sec:fast_sampling}
Let $\cbraces{v_{m}}_{m=0}^{N_{\os}}\subset \mbbV$ be the sequence of control vectors defining $\bar{v}_{\os} \in\mcalP\mcalC\paren{\brackets{ 0, \deltat},\mbbV}$ and $\deltatos \coloneqq \deltat/N_{\os}$. 
On the oversampling grid, we have the state vectors $\x{m}\coloneqq \xc{m\,\deltatos}$ and the discrete time models $\x{m+1}=\fos{ { \x{m} , v_{m} } }$ and $\Fos{x_{m} }\,\omega_{m}\coloneqq \sum_{i=1}^{\abs{\Omega}} \fos{ { x_{m}, v^{i} } }\,\omega_{k,i}=\fos{x_{m},v_{m} }$, which are defined analogously to \eqref{eq:sys_disc} and \eqref{eq:sys_disc_convex} for all $m\in[N_{\os}]$. 
Let $\omegas\coloneqq \brackets{\omega_{0},\omega_{1},\ldots,\omega_{N_{\os}-1}}\in (\mbbS^{\abs{\Omega}})^{N_{\os}}$ be a stacked sequence of integer-feasible convex multipliers. 
We denote the recursive solution to the outer convexified system dynamics on the oversampling grid by $\xsolos{m,x,\omegas} \coloneqq \Fos{\xsolos[eq]{m-1,x, \omegas}}\,\omega_{m-1}$ for all $m\in\cbraces{1,2,\ldots,N_{\os}}$ with $\xsolos{0,x,\omegas} \coloneqq x$.
If $\omegas = \brackets{u,u,\ldots,u}\in(\mbbU^{\abs{\Omega}})^{N_{\os}}$, we have that $\xsolos{N_{\os},x,\omegas} = \F{x}\,u$. 
The time grid with the step width $\deltat$ mainly determines the number of optimization variables in the transcribed version of DT-OCP~\eqref{eq:dt-ocp} and ultimately the computational load. 
We only introduce the finer time grid to perform rounding with frequent switching after the completion of the optimization process underlying DT-OCP~\eqref{eq:dt-ocp}. The objective of the downstream fast rounding phase is to reduce the integer approximation error.


\section{Practical Set-Point Stabilization}
\label{sec:set_point_stabilization}

Using the optimal solutions to DT-OCP in~\eqref{eq:dt-ocp} for closed-loop control of system \eqref{eq:sys_disc_convex} on time intervals $[n\deltat,(n+1)\deltat)$ with $n\in\mbbN_{0}$,  it operates as an autonomous system for which we need to ensure asymptotic stability properties. 
We aim to asymptotically stabilize the steady-state $\paren[in]{\xf ,\uf}$ in $\feasibleset$ for the following relaxed and autonomous system:
\begin{equation}
\x{n+1} = \F{\x{n}}\,\xmu{\x{n}} ,\,\xmu{\x{n}}\coloneqq\uopt{0}(\x{n}) \in \mbbU^{\abs{\Omega}}.
\label{eq:sys_cl}
\end{equation}
Let the recursive solution to system \eqref{eq:sys_cl} be denoted by $\xsolcl{n,x}$ for all $n\in\mbbN_{0}$ with $\xsolcl{0,x}\coloneqq x$.
We also want to transfer the binary-feasible and autonomous system 
 \begin{equation}
\xd{n+1} \coloneqq \xsolos{N_{\os},\xd{n},\omegasmu{\xd{n}}}, \omegasmu{\xd{n}} \in (\mbbS^{\abs{\Omega}})^{N_{\os}},
 \label{eq:sys_cl_d}
 \end{equation}
in some small neighborhood of the steady-state for all $\xd{n}\in \levVN{\rho}\subseteq \feasibleset$, where it shall remain for all time steps. 
Here, $\omegasmu{x}\coloneqq [{\omega_{\mu}}_{0}(x),{\omega_{\mu}}_{1}(x),\hdots,{\omega_{\mu}}_{N_{\os}-1}(x)] \in (\mbbS^{\abs{\Omega}})^{N_{\os}} $ denotes the state-dependent control law that shall be derived from $\xmu{x}\in\mbbU^{\abs{\Omega}}$ via integer reconstruction (see Sec. \ref{sec:input_rounding}) on the oversampling grid from \mbox{Section~\ref{sec:fast_sampling}}.
According to \cite{Allan2017,Rawlings2020}, we assume that $\levVN{\rho}$ with some $\rho>0$ is a compact sublevel set of the optimal value function $\VN{}$ that is contained in the closed feasible set $\feasibleset$.
Let the recursive solution to system \eqref{eq:sys_cl_d} be denoted by $\xsolcld{n,x}$ for all $n\in\mbbN_{0}$ with $\xsolcld{0,x}\coloneqq x$.
We initialize both closed-loop systems with some feasible initial state $\xd{0}=\xo=x\in\levVN{\rho}$. 

\subsection{Nominal Stability with Stabilizing Terminal Ingredients}
The following Assumption \ref{assum:terminal_control_law} ensures both recursive feasibility of the terminal set $\terminalset$ and asymptotic stability of the origin in $\terminalset$ for the closed-loop system \eqref{eq:sys_cl}. 

\begin{assum}[Control Invariant Terminal Set]\label{assum:terminal_control_law}
There exists a stabilizing terminal control law $\xmuf{}: \terminalset \to \mbbU^{\abs{\Omega}}$ such that for all $x\in\terminalset$ it holds that:
\begin{align}
\F{x}\,\xmuf{x} &\in \terminalset,\\
\Jfc{\F{x}\,\xmuf{x}} - \Jfc{x} &\leq -\xellc{x} - \uellc{\xmuf{x}}.
\end{align}
\end{assum}

Conventional MPC with stabilizing terminal conditions inherits the stabilizing properties of $\xmuf{}$ in $\terminalset$ and transfers them to the feasible state space $\feasibleset$ \cite{Rawlings2020,Mayne2000}. 
In \cite{Rawlings2020,Chen1998}, the authors present a possible procedure for determining a control invariant terminal set $\terminalset=\levJf{\pi}$.

\begin{prop}[Asymptotic Stability After Relaxation]\label{prop:nominal_stability}
Suppose Assumptions~\ref{assum:steady-state}--\ref{assum:terminal_control_law} hold. Then the optimal value function~$\VN{}$ is a Lyapunov function on the feasible set~$\feasibleset$:  
\begin{align}
&\alphaonenorm{x} \leq  \VN{x}\leq\alphatwonorm{x},~\alphaone{},\alphatwo{}\in\mcalK_{\infty},\\
&\VN{\F{x}\,\xmu{x}}  \leq \VN{x} -\alphaonenorm{x},~\forall\,x \in \feasibleset.
\end{align}
The origin is asymptotically stable in the positive invariant set~$\feasibleset$ for the relaxed system \eqref{eq:sys_cl}. There is a function $\beta \in \mcalK\mcalL$
such that for all $x \in \feasibleset$ and $n\in \mbbN_{0}$, we have that
$\norm{\xsolcl{n,x}}\leq \beta\paren{\norm{x},n}$.
\end{prop}

\subsection{Inherent Robustness to Input Perturbations}

Restoring integer feasibility implies the need to round the relaxed control inputs. The deviation from the optimal control value can be interpreted as an input perturbation.
Before we specify the integer-feasible control law $\omegasmu{}$ more precisely, we want to integrate general input perturbed systems with the inherent robustness properties of nominal MPC with stabilizing terminal conditions due to \cite{Yu2014,Allan2017}.

The authors in \cite{Allan2017,Yu2014} show that nominal MPC is inherently robust to small state disturbances and estimation errors.
Allan et al. \cite{Allan2017} demonstrate robust asymptotic stability~(RAS) based on input-to-state stability~(ISS) for suboptimal MPC, which is founded on the continuous cost function $\JN{}$ and an extended state containing a systematically warm-started control sequence.
To capture the unknown outcome of the suboptimal solver, the authors analyze the evolution of the closed-loop system governed by a difference inclusion. 
However, the theoretical findings obtained under suboptimal conditions a general enough to cover the optimal MPC performance. 
Building on this observation and motivated by the set-point stabilization under plant-model mismatch, the authors in~\cite{Kuntz2025} deduce inherent robustness results in the presence of parameter errors. 
The authors employ the optimal value function $\VN{}$ as the Lyapunov candidate and note that their ISS and RAS results follow a special case of the inherent robustness results of suboptimal MPC in~\cite{Allan2017}.

To deduce an established closed-loop stability condition in our relaxed domain with minimal effort, we build upon the more constructive results in~\cite[Chap. 3.2.4]{Rawlings2020}, which are similar to those for continuous-time systems in~\cite{Yu2014}.
The authors show robust positive invariance and a sufficient cost decay to reach some small and positive invariant neighborhood of the steady-state in the presence of bounded disturbances. 
We formally integrate the derivations from~\cite{Allan2017,Rawlings2020} with the definition of $\mcalP$-practical asymptotic stability due to~\cite{Gruene2017}. 
We deduce $\mcalP$-practical asymptotic stability as a special case of the ISS/RAS results in \cite{Allan2017} and use an arbitrarily small level set $\levVN{\kappa}$ as the positive invariant neighborhood of the origin. 

\begin{prop}[Stability With Input Perturbations]\label{prop:practical_stability}
Suppose Assumptions \ref{assum:steady-state}--\ref{assum:terminal_control_law} hold. 
For every $ \kappa \in(0,\rho) $ with $\mcalP\coloneqq\levVN{\kappa}$, there exists a constant $\gamma>0$ such that if
\begin{equation}
\norm{\xsolos{N_{\os},x,\omegasmu{x}} -\F{x}\,\xmu{x}}\leq \gamma
\label{eq:practical_stability_state_bound}
\end{equation}
for all $x\in\levVN{\rho}\subseteq\feasibleset$, then $\levVN{\rho}$ and $\mcalP$ are positive invariant sets for system~\eqref{eq:sys_cl_d}. The optimal value function $\VN{}$ is a Lyapunov function on $ \levVN{\rho}\backslash\mcalP$. For all $x\in\levVN{\rho}\backslash\mcalP$, we have that:
\begin{align}
&\alphaonenorm{x} \leq  \VN{x}\leq\alphatwonorm{x},~\alphaone{},\alphatwo{},\alpha_{3}\in\mcalK_{\infty}, \\
&\VN{\xsolos{N_{\os},x,\omegasmu{x}}} \leq \VN{x} -\alphathreenorm{x}.
\end{align}
The origin is $\mcalP$-practically asymptotically stable in the positive invariant set $\levVN{\rho}$ for system \eqref{eq:sys_cl_d}. 
Therefore, there exists a function $\beta \in \mcalK \mcalL$ such that for all $x \in \levVN{\rho}$ and $n\in\mbbN_{0}$ with $\xsolcld{n,x}\not \in \mcalP$, we have that
$\norm{ \xsolcld{n,x}} \leq \beta\paren{\norm{x},n}$.
\end{prop}


\section{Control Input Rounding}
\label{sec:input_rounding}

The idea is to derive a suboptimal but integer feasible control sequence $\omegasmu{x}$ from the relaxed control law $\xmu{x}$ for all $x \in \levVN{\rho}\subseteq \feasibleset$. 
From \cite{Sager2010,Kirches2020}, we extract that the state approximation error is linearly dependent on the integrated control rounding error. 
To apply the results from \cite{Sager2010,Kirches2020}, we require mild assumptions about the original system properties.
 \begin{assum}[Lipschitz Continuity]
The vector field $f$ is Lipschitz continuous in its first argument on every compact subset of $\feasiblesetc$ with Lipschitz constant $L>0$.
\label{assum:lipschitz} 
\end{assum}
 \begin{assum}[Differentiability]
The continuous mapping $t\mapsto \fc{ { \xc{t}, v^{i} } }$ is differentiable a.e. and its derivative admits an upper bound $C>0$ such that
\begin{equation}
\norm[eq]{\frac{\mathrm{d}}{\dt} \fc[eq]{ { \xc{t}, v^{i} } }} \leq C
\end{equation}
holds for all $v^{i}\in\Omega $ and $t \in[0,\tf]$ a.e. with $\xc{t}\in\feasiblesetc$. 
\label{assum:continuous_diff} 
\end{assum}
 \begin{assum}[Boundedness]
The continuous mapping $t\mapsto \fc{ { \xc{t}, v^{i} } }$ is bounded by some $M>0$ such that 
\begin{equation}
\norm[eq]{ \fc[eq]{ { \xc{t}, v^{i} } }} \leq M
\end{equation}
\label{assum:boundedness} 
holds for all $v^{i}\in\Omega $ and $t\in[0,\tf]$ with $\xc{t}\in\feasiblesetc$. 
\end{assum}

 In the following, we tailor the continuous-time formulation in \cite{Sager2010,Kirches2020} to our discrete-time formulation. 
 \begin{prop}[Upper Approximation Bound, see \cite{Sager2010}]\label{prop:sur_state_bound}
Suppose that Assumptions \ref{assum:lipschitz}--\ref{assum:boundedness} hold. 
Assume that the control law $\xmud{x,m,t}\coloneqq\sum_{i=0}^{m-1}{\omega_{\mu}}_{i}(x)\, \mathds{1}_{[i\deltatos,(i+1)\deltatos)}(t)$ satisfies the following integral pseudo metric for all $x\in\feasibleset$ with $\sigma>0$:
\begin{equation}
\sup_{t\in[0,\deltat]} \norm[eq]{\int_{0}^{t} \xmu{x} - \xmud{x,N_{\os},\tau}\dtau}  \leq \sigma.
 \label{eq:sur_input_bound}
\end{equation}
Then for all $x\in\feasibleset$, we obtain that:
\begin{equation}
\norm{\xsolos{N_{\os},x,\omegasmu{x}} -\F{x}\,\xmu{x}}\leq (M+C\deltat) \,\sigma\, \mathrm{e}^{L\deltat}.
\label{eq:sur_state_bound}
\end{equation}
 \end{prop}

In Proposition~\ref{prop:practical_stability}, we claim that there exists some bound $\gamma>0$ for the state deviation caused by input perturbations that ensures $\mcalP$-practical asymptotic stability,  see 
\eqref{eq:practical_stability_state_bound}. 
In~\eqref{eq:sur_state_bound}, we specify an upper state error bound that formally depends on the integrated input error in \eqref{eq:sur_input_bound}. 
Notice that if $\sigma \to 0$, the state error bound in \eqref{eq:sur_state_bound} also tends to zero. However, the latter observation does not imply point-wise convergence of $\xmud{x,N_{\os},\cdot}$ and $\xmu{x}$ in time. Therefore, Proposition \ref{prop:sur_state_bound}  motivates systematic and fast switching in the context of input rounding such that the right-hand side of \eqref{eq:sur_state_bound} becomes smaller than $\gamma$ in \eqref{eq:practical_stability_state_bound}.
Notice that the upper state approximation bound in \eqref{eq:sur_state_bound} does not yet depend on the specific rounding algorithm, it mainly results from Grönwall's Lemma, see \cite[Thm. 2]{Sager2010}. 

\textit{Sum-Up Rounding (SUR)}.
In the following, we present possible and SOS1 admissible rounding algorithms to restore integer feasibility in polynomial time, namely simple rounding (SR) and sum-up rounding (SUR), see~\cite{Sager2010,Kirches2020}. Alternative optimization based rounding approaches, which are more complex and rely, inter alia, on integer linear programming, are presented in~\cite{Sager2011,Jung14,Bestehorn2019}.

Similar to \cite{Kirches2020}, we define control rounding as follows:
\begin{equation}
{\omega_{\mu}}_{m}(x)\coloneqq\mbf{1}^{i^{*}(m,x)},~\forall\,m\in [N_{\os}],
\label{eq:sur_control_law}
\end{equation}
where the index $i^{*}(m,x)$ is the solution to the problem:
\begin{equation}
i^{*}(m,x)\coloneqq\argmax_{i\in\cbraces{1,2,\ldots,\abs{\Omega}}} \cbraces[eq]{\eta^{\cbraces{\mathrm{SR,SUR}}}_{i}(m,x)}.
\label{eq:opt_index}
\end{equation}

Ties may be broken arbitrarily.
In the case of SR, we define $\eta^{\mathrm{SR}}_{i}(m,x)\coloneqq \mu_{i}(x) \in [0,1]$.
We detect the largest input value in each dimension $i\in\cbraces{1,2,\ldots,\abs{\Omega}}$ and apply SOS1 admissible rounding on the oversampling time grid. 
Since SR repeats its decision $N_{\os}$ times for piecewise constant controls, the error bound in~\eqref{eq:sur_input_bound} is $\sigma=\sigma^{\mathrm{SR}} \coloneqq N_{\mathrm{os}}\,\deltatos\, \sqrt{1-\abs{\Omega}^{-1}}$.

In case of SUR, the inner integral argument is defined by:
\begin{equation}
\eta^{\mathrm{SUR}}_{i}(m,x)\coloneqq \int_{0}^{(m+1) \deltatos} \mu_{i}(x)\dt - \int_{0}^{m \deltatos} {\mu_{\mathrm d}}_{i}(x,m,t)\dt.
\label{eq:integral_argument}
\end{equation}
In contrast to SR, we determine the largest input value in each dimension $i\in\cbraces{1,2,\ldots,\abs{\Omega}}$ and on each time interval $[0,(m+1)\deltatos)$ with $m\in [N_{\os}]$, taking into account all previous rounding decisions. 
The tightest upper rounding error bound for SUR in \eqref{eq:sur_input_bound} is given by \cite{Kirches2020}:
\begin{equation}
\sigma=\sigma^{\mathrm{SUR}}\coloneqq\sqrt{\abs{\Omega}}\,\deltatos\sum_{i=2}^{\min\cbraces{\abs{\Omega},N_{\mathrm{os}}+1}}\frac{1}{i}.
 \label{eq:sur_tightest_bound}
\end{equation}
The factor $\sqrt{\abs{\Omega}}$ describes the equivalence relation between the maximum norm used in~\cite{Kirches2020} and the Euclidean norm. 
The input approximation error for SUR reaches its maximum value after $\min\cbraces{\abs{\Omega},N_{\mathrm{os}}+1}$ rounding steps, see \cite[Thm. 6.1]{Kirches2020}. 
The upper bound in \eqref{eq:sur_tightest_bound} does not scale with the number of oversampling steps $N_{\os}$.
In contrast to SR, we obtain that $\sigma \to 0$ if $\deltatos \to 0$. 
The latter observation brings us to our main result.
\begin{thm}[Fast Switching in MI-MPC]
Suppose Assumptions~\ref{assum:steady-state}--\ref{assum:boundedness} hold. 
Assume $\omegasmu{x}$ follows from SUR and is defined by \eqref{eq:sur_control_law}--\eqref{eq:integral_argument} for all $x \in \levVN{\rho}\subseteq\feasibleset$. 
For every $\kappa \in (0,\rho)$ with $\mcalP\coloneqq\levVN{\kappa}$, there exists an oversampling step width $\deltatos>0$ such that  the origin is $\mcalP$-practically asymptotically stable in the positive invariant set $\levVN{\rho}$ for the input perturbed system \eqref{eq:sys_cl_d}.
\label{thm:theorem_1}
\end{thm}

\begin{proof}
Let $\xdplus\coloneqq \xsolos{N_{\os},x,\omegasmu{x}}$ and $\xplus\coloneqq \F{x}\,\xmu{x}$.
We combine \eqref{eq:sur_tightest_bound} with \eqref{eq:sur_state_bound} and obtain:
\begin{equation}
\norm[eq]{\xdplus-\xplus }\leq (M+C\deltat) \,\sqrt{\abs{\Omega}}\,\deltatos \textstyle\sum_{j=2}^{\abs{\Omega} }\frac{1}{j}\, \mathrm{e}^{L\deltat}.
\end{equation}
Now, we enforce the right-hand side to be smaller than the upper bound $\gamma>0$ from Proposition~\ref{prop:practical_stability}:
\begin{equation}
(M+C\deltat) \,\sqrt{\abs{\Omega}}\,\deltatos \textstyle\sum_{j=2}^{\abs{\Omega} }\frac{1}{j}\, \mathrm{e}^{L\deltat}\leq \gamma.
\end{equation}
Finally, we rearrange the inequality to $\deltatos$ and impose that:
\begin{equation}
\deltatos \leq\delta t_{\max}\coloneqq \frac{\gamma}{(M+C\deltat) \,\sqrt{\abs{\Omega}}\sum_{j=2}^{\abs{\Omega} }\frac{1}{j}\,\mathrm{e}^{L\deltat}}>0.
\label{eq:upper_bound_deltat}
\end{equation}
For every $\kappa \in (0,\rho)$, there exists an arbitrary small $\deltatos$ such that $\norm[eq]{\xdplus-\xplus }\leq \gamma$ holds for all $x\in\levVN{\rho}$.
$\mcalP$-practical asymptotic stability thus follows from Proposition~\ref{prop:practical_stability}.
\end{proof}

In theory, we can use an arbitrary small sampling width~$\deltatos$ to satisfy the upper bound $\delta t_{\max}$. 
In practice, there usually also exists a lower bound $ \deltatos_{\min}\in(0,\deltatos_{\max})$ that is determined by technical properties, such as transient reversal processes or safety functions. 
However, in electrical power systems, e.g., switching times can be assumed to be very small \cite{Karamanakos2020}.


\section{Numerical Experiments}
\label{sec:simulation}

We investigate the Van-der-Pol Oscillator with state dimension $\nx = 2$ and a nonlinear input with dimension $\nv=1$:
\begin{equation}
\fc[eq]{\xc{t},\vc{t}} =\begin{pmatrix}
\bar x_{2}(t)\\
\paren[eq]{1-\bar x_{1}^{2}(t)}\,\bar x_{2}(t)-\bar x_{1}(t) + \sin\paren[eq]{\vc{t}}
\end{pmatrix}.
\label{eq:vdp}
\end{equation}
Let us consider quadratic cost functions $\xellc{ x } = x^{\intercal}Qx$, $\vellc{v} = v^{\intercal}R_{v}v$,  and $\uellc{u} = \sum_{i=1}^{\abs{\Omega}}\vellc{ {v^{i} } }\,\paren[eq]{{u_{i}-\ufi}}^{2}$ with $\uf = (0.5,0.5)^{\intercal}$. The matrices $Q$ and $R_{v}$ denote positive definite weighting matrices.
The terminal cost function $\Jfc{x} = x^{\intercal}Px$ approximates the infinite-horizon costs of the nonlinear system, where $P\in \mbbR^{\nx \times \nx}$ is the positive definite solution to the discrete-time algebraic Riccati equation $A^{\intercal} P A - (A^{\intercal} P B) (\varrho W+B^{\intercal}PB)^{-1}(B^{\intercal}PA)+\varrho Q$ with $W\coloneqq \mathrm{diag}(R)=\mathrm{diag}(\vellc{v^1},\ldots,\vellc{v^{\abs{\Omega}}}).$
The pair $(A,B)$ denotes the stabilizable linearization of the relaxed system $\F{ x }\,u$ at the integer infeasible and unstable steady-state $(\xf,\uf) = (0,\uf)$. 
Without a proof, we choose $\pi = 0.3$ and $\varrho=1.001$ according to the setup procedure in \cite{Rawlings2020}. 
We define the finite control set as $\Omega= \cbraces{-1,1}\subset \mbbV = \set{v\in\mbbR}{\abs{v}\leq 1} $. 
With $Q=I$ and $R_{v}=1$, we have that $W=I$. 
For the time discretization, we choose $\deltat = 0.15\,\mathrm{s}$ and $N=20$.
We apply direct collocation (midpoint method) on the finest time grid with $\deltatos = 0.005\,\mathrm{s}$ to transcribe the DT-OCP in \eqref{eq:dt-ocp} into a NLP. However, for optimization, the control is always discretized on the coarse time grid with $\deltat = 0.15\,\mathrm{s}$.
The numerical benchmark setup is built upon the automatic differentiation and optimization framework CasADi~\cite{Andersson2019}, the general purpose solver IPOPT \cite{Waechter2005}, and the sparse linear solver MUMPS \cite{MUMPS:1}. 

\newlength{\figureheight} 
\newlength{\figurewidth}  

\begin{figure}[t]
	\centering
	\vspace{\topsep}
	\setlength\figureheight{0.75\columnwidth} 
	\setlength\figurewidth{0.95\columnwidth} 
\begin{tikzpicture}

\begin{axis}[
	width=0.36\figurewidth,
	height=0.32\figurewidth,
	axis y line = left,
	axis x line = bottom,
	xticklabel style = {/pgf/number format/assume math mode = true, /pgf/number format/fixed},
	yticklabel style = {/pgf/number format/assume math mode = true, /pgf/number format/fixed},
	at={(0\figurewidth,0.5\figurewidth)},
	scale only axis,
	separate axis lines,
	every outer x axis line/.append style={black},
	every x tick label/.append style={font=\color{black},font=\footnotesize},
	xmin=-0.05,
	xmax=0.505,
	xmajorgrids,
	xtick={-0.05,0,0.1,0.2,0.3,0.4,0.5},
	xticklabels={,0,0.1,0.2,0.3,0.4,0.5},
	x label style={font=\footnotesize},
	y label style={at={(-0.22,0.5)},font=\footnotesize},
	every outer y axis line/.append style={black},
	every y tick label/.append style={font=\color{black},font=\footnotesize},
	ymin=-0.31,
	ymax=0.15,
	line join=round,
	xlabel={${\phi_{\deltat_1}}(\cdot,x)$,~${\check\phi_{\deltat_1}}(\cdot,x)$},
	ylabel={${\phi_{\deltat_2}}(\cdot,x)$,~${\check\phi_{\deltat_2}}(\cdot,x)$},
	ymajorgrids,
        ytick={-0.3,-0.2,-0.1,0,0.1,0.15},
        yticklabels={-0.3,-0.2,-0.1,0,0.1,},
	legend columns = 3,
	legend entries={Relaxed, SUR $\deltatos = \deltat$ (SR), SUR $\deltatos = \deltat/10$,},
	legend style={at={(-0.1,1.02)},anchor=south west,legend cell align=left,align=left,draw=none,font=\footnotesize,fill=none}
	]

\addlegendimage{color=IEEEblack,solid,line width=1.2pt};
\addlegendimage{color=IEEEgray,solid,line width=1.2pt};
\addlegendimage{color=IEEEblue,solid,line width=1.2pt};

\addplot[color=IEEEgray,solid,line width=1pt]
table[]{content/data/sur_state_mb_1.txt};

\addplot[color=IEEEblue,solid,line width=1pt]
table[]{content/data/sur_state_mb_10.txt};

\addplot[color=IEEEblack,solid,line width=1pt]
table[]{content/data/nr_state.txt};

\end{axis}

\begin{axis}[
	width=0.36\figurewidth,
	height=0.32\figurewidth,
	axis y line = left,
	axis x line = bottom,
	xticklabel style = {/pgf/number format/assume math mode = true, /pgf/number format/fixed},
	yticklabel style = {/pgf/number format/assume math mode = true, /pgf/number format/fixed},
	at={(0.5\figurewidth,0.5\figurewidth)},
	scale only axis,
	separate axis lines,
	every outer x axis line/.append style={black},
	every x tick label/.append style={font=\color{black},font=\footnotesize},
	xmin=-0.05,
	xmax=0.1,
	xmajorgrids,
	xtick={-0.05,-0.025,0,0.025,0.05,0.075,0.1},
	xticklabels={-0.05,,0,,0.05,,0.1},
	x label style={font=\footnotesize},
	y label style={at={(-0.18,0.5)},font=\footnotesize},
	every outer y axis line/.append style={black},
	every y tick label/.append style={font=\color{black},font=\footnotesize},
	ymin=-0.1,
	ymax=0.05,
	line join=round,
        xlabel={${\phi_{\deltat_1}}(\cdot,x)$,~${\check\phi_{\deltat_1}}(\cdot,x)$},
	ymajorgrids,
	ytick={-0.1,-0.075,-0.05,-0.025,0,0.025,0.05},
	yticklabels={-0.1,,-0.05,,0,,0.05},
	legend columns = 0,
	legend entries={},
	]
	
\addplot[color=IEEEgray,solid,line width=0.8pt]
table[]{content/data/sur_state_mb_1.txt};

\addplot[color=IEEEblue,solid,line width=0.8pt]
table[]{content/data/sur_state_mb_10.txt};

\addplot[color=IEEEblack,solid,line width=1pt]
table[]{content/data/nr_state.txt};

\end{axis}

\begin{axis}[
	width=0.36\figurewidth,
	height=0.32\figurewidth,
	axis y line = left,
	axis x line = bottom,
	at={(0\figurewidth,0.0\figurewidth)},
	scale only axis,
	separate axis lines,
	every outer x axis line/.append style={black},
	every x tick label/.append style={font=\color{black},font=\footnotesize},
	xmin=-0.02,
	xmax=7.5,
	xmajorgrids,
	x label style={font=\footnotesize},
	y label style={at={(-0.22,0.5)},font=\footnotesize},
	every outer y axis line/.append style={black},
	every y tick label/.append style={font=\color{black},font=\footnotesize},
	ymin=-0.025,
	ymax=1.025,
	line join=round,
	xlabel={$n\,\deltat~[\mathrm s]$},
	ylabel={$\xmu{\x{n}},~\xmud{\xd{n},1,\cdot}$},
	ytick = {0,0.125,0.25,0.375,0.5,0.625,0.75,0.875,1,1.025},
	yticklabels = {$0$,,$0.25$,,$0.5$,,$0.75$,,$1$,},
	xtick = {0,1.25,2.5,3.75,5,6.25,7.5},
	xticklabels = {$0$,,$2.5$,,$5$,,$7.5$},
	ymajorgrids,
	legend image post style={xscale=0.5},
	legend columns = 2,
	legend entries={$ v^{1}=-1$,$v^{2}=+1$},
	legend style={at={(-0.03,0.97)},anchor=south west,legend cell align=left,align=left,draw=none,font=\scriptsize,fill=none}
	]

\addplot[forget plot,const plot, color=IEEEgray,solid,line width=1pt]
table[x index=1, y index = 2]{content/data/sur_control_mb_1.txt};

\addplot[forget plot, const plot, color=Darkgray,densely dashed,line width=1pt]
table[x index=1, y index = 3]{content/data/sur_control_mb_1.txt};

\addplot[const plot, color=IEEEblack ,solid,line width=1pt]
table[x index=1, y index = 2]{content/data/nr_control.txt};
	
\addplot[const plot, color=IEEEblack,densely dashed,line width=1pt]
table[x index=1, y index = 3]{content/data/nr_control.txt};

\end{axis}

\begin{axis}[
	width=0.36\figurewidth,
	height=0.32\figurewidth,
		xticklabel style = {/pgf/number format/assume math mode = true, /pgf/number format/fixed},
	yticklabel style = {/pgf/number format/assume math mode = true, /pgf/number format/fixed},
	axis y line = left,
	axis x line = bottom,
	at={(0.5\figurewidth,0\figurewidth)},
	scale only axis,
	separate axis lines,
	every outer x axis line/.append style={black},
	every x tick label/.append style={font=\color{black},font=\footnotesize},
	xmin=0,
	xmax=10,
	xmajorgrids,
	x label style={font=\footnotesize},
	y label style={at={(-0.15,0.5)},font=\footnotesize},
	every outer y axis line/.append style={black},
	every y tick label/.append style={font=\color{black},font=\footnotesize},
	ymin=-0.02,
	ymax=5,
	line join=round,
	xtick={0,2.5,5,7.5,10,12.5,15},
	ytick={0,1,2,3,4,5},
        xlabel={$n\,\deltat~[\mathrm s]$},
	ylabel={$\VN{\x{n}},\,\VN{\xd{n}}$},
	ymajorgrids,
	legend columns = 1,
	legend entries={$\deltatos = \deltat$, $\deltatos = \deltat/2$,$\deltatos = \deltat/5$,$\deltatos = \deltat/10$,Relaxed},
	legend style={at={(0.2,0.28)},anchor=south west,legend cell align=left,align=left,draw=none,font=\scriptsize,fill=none}
	]

\addplot[color=IEEEgray,solid,line width=1pt]
table[x index=1, y index = 2]{content/data/sur_costs_mb_1.txt};

\addplot[color=green,solid,line width=1pt]
table[x index=1, y index = 2]{content/data/sur_costs_mb_2.txt};

\addplot[color=brown,solid,line width=1pt]
table[x index=1, y index = 2]{content/data/sur_costs_mb_5.txt};

\addplot[color=IEEEblue,solid,line width=0.8pt]
table[x index=1, y index = 2]{content/data/sur_costs_mb_10.txt};
	
\addplot[color=IEEEblack,dashed,line width=0.8pt]
table[x index=1, y index = 2]{content/data/nr_costs.txt};
	
\end{axis}

\end{tikzpicture}%
	\caption{Set-point stabilization for systems \eqref{eq:sys_cl} and 
	\eqref{eq:sys_cl_d} with SUR. 
	Top left: Phase portrait. Top right: Close-up view of left plot. 
	Bottom left: Closed-loop control trajectories of systems \eqref{eq:sys_cl} and 
	\eqref{eq:sys_cl_d} with $\deltatos=\deltat$.
	Bottom right: Evolutions of the optimal value function $\VN{}$ for different sampling times $\deltatos$. }
	\label{fig:figure}
	\vspace{-\topsep}
\end{figure}

In  Figure~\ref{fig:figure}, the implicit control law $\xmu{}$ transfers, as expected, the relaxed system \eqref{eq:sys_cl} from $x = (0.5,0)^{\intercal}$ to the origin and then stabilizes the steady-state $((0,0)^{\intercal},(0.5,0.5)^{\intercal})$ according to Proposition~\ref{prop:nominal_stability}. 
Note that for $\deltatos = \deltat$, SUR equals SR, taking a single rounding decision on the interval $\deltat$ (see bottom left plot of Fig.~\ref{fig:figure}). 
The perturbed control law $\xmud{}$ distracts the relaxed system~\eqref{eq:sys_cl_d} at an early stage in time and evolves into a oscillating state space behavior.
The evolution of the optimal value function in the bottom right plot of Figure~\ref{fig:figure} reveals that costs do not decrease in the sense of Lyapunov. 

For SUR with $\deltatos=\deltat/10=0.015\,\mathrm{s}$, we observe a state trajectory that is on average much more similar to the relaxed state trajectory.
The close-up of the state-space in the upper right plot shows that the relaxed system remains inside some small neighborhood of the origin for all time steps. 
Notice that for $\deltatos=\deltat/10$, the evolution of the optimale value function nearly resembles the evolution of the relaxed optimal value function, which decreases in the sense of Lyapunov. 

Let us consider Table \ref{tab:table}, where we evaluate the maximum control integer gap $\sigma_{\max}\coloneqq\max_{n\in[120]}\cbraces{ \norm{\int_{0}^{\deltat} \xmu{\xd{n}} - \xmud{\xd{n},N_{\os},t}\dt}}$ and the maximum state gap $\gamma_{\max}\coloneqq \max_{n\in [120]}\norm{\xsolos{N_{\os},\xd{n},\omegasmu{\xd{n}}} -\F{\xd{n}}\,\xmu{\xd{n}}}$ for the numerical setup as in Figure~\ref{fig:figure}. 
We define the relative run time $t_{r} \in (0\,\%,100\,\%]$ as the quotient of the minimum run time of SUR and the minimum run time of the optimization over 120 steps. 
We apply warm-starting of primal und dual variables.

The finer the time resolution of the oversampling grid, the smaller is the control integer gap $\sigma_{\max}$ for SUR. 
This observation would not hold for SR as it simply repeats its decision.
For SUR, the state and control integer gaps almost vanish for $\deltatos =\deltat/30$ according to \eqref{eq:sur_state_bound} and \eqref{eq:sur_tightest_bound}. Hence, we know that for $\deltatos \to 0$, we have that $\gamma \to 0$ and thus also $\kappa \to 0$. For $\deltatos=\deltat/30$, the resulting state trajectory is nearly identical to that of the relaxed trajectory since $\gamma_{\max}=0.0024$. 
Therefore, we do not plot this graph in Figure~\ref{fig:figure}. 
These numeric results support our claim in Theorem~\ref{thm:theorem_1}.
The relative run time evaluation confirms a negligible amount of time for integer reconstruction after optimization with $t_{\mathrm r} < 1\,\%$.

\begin{rem}[Quadratic Convex Multipliers]
Squaring the control variables contradicts the linear cost term in~\eqref{eq:linear_cost}, which is important to maintain the bijective mapping to $l$ in \eqref{eq:orig_cost}. Therefore, we loose optimality claims with respect to the originally outer convexified problem~\eqref{eq:dt-ocp}, even if SUR would enable an error-free reconstruction. 
Nevertheless, since we choose an integer-infeasible steady-state, we have that $\uf\not\in\mbbS^{\abs{\Omega}}$. We thus have to rely on the cost regularization term in~\eqref{eq:abs_mult} and not in~\eqref{eq:linear_cost}.
Since we focus on closed-loop stability and assume that fast switching is technically possible and inexpensive, we do not have to ensure the bijective cost mapping.
It is sufficient to have any relaxed and asymptotically stable closed-loop system evolution $\xsolcl{\cdot,x}$ to which we can attach the evolution of the input perturbed closed-loop system $\xsolcld{\cdot,x}$.
\end{rem}


\section{Conclusion}

\begin{table}[t]
\vspace{\topsep}
    \caption{Maximum control integer and state gaps with SUR.}
    \centering
    \begin{tabular}{ccccccc} 
        \toprule
          & $\deltatos=$& $\deltat$ & $\deltat/2$ & $\deltat/5$ & $\deltat /10$& $\deltat /30$ \\ 
        \midrule
         \multirow{2}{*}{SUR}  &$\sigma_{\max}$  & 0.1059  & 0.0525  & 0.0207 & 0.0105 & 0.0035\\
                                           &$\gamma_{\max}$ & 0.1036  & 0.0411  & 0.0173 & 0.0113 & \textbf{0.0024}\\
                                           &$t_{\mathrm{r}}[\%]$ & 0.1563  & 0.1870  & 0.1875 & 0.2444 & \textbf{0.2856}\\
        \bottomrule
    \end{tabular}
    \vspace{-\topsep}
    \label{tab:table}
\end{table}

We have formally deduced $\mcalP$-practical asymptotic stability for mixed-integer model predictive control. To achieve this stability results, we have fused the findings on inherent robustness of conventional model predictive control (closed-loop) with stabilizing terminal conditions due to \cite{Mayne2000,Allan2017,Rawlings2020} with the integer approximation algorithm for mixed-integer optimal control (open-loop) due to \cite{Sager2007,Sager2010} and its theoretical properties in \cite{Sager2010,Kirches2020} on the boundedness of the approximation errors. 

After applying partial outer convexification and relaxation, we obtain relaxed system and cost formulations. This relaxed optimal control problem serves as the foundation for providing a desired stabilizing state space behavior for the real and integer feasible system. 
We have introduced sum-up rounding of the implicit control law on an oversampled temporal grid. 
If we assume that switching is inexpensive and technically possible, we have shown theoretically and with a numerical example that the approximation error due to rounding tends to zero the faster we switch the control input. 
This approximation property is especially interesting for applications with fast dynamics such as power electronics. 


\begin{appendix}
\label{sec:appendix}

\renewcommand{\thethm}{A\arabic{thm}}

\subsection{Proof of Proposition~\ref{prop:nominal_stability}}

\begin{proof}
The functions $\xellc{}$, $\uellc{}$, $\Jfc{}$, and $\F{}$ are assumed to be continuous. 
Assumptions \ref{assum:steady-state}, \ref{assum:constraint_sets}, and \eqref{eq:abs_mult} ensure proper steady-state conditions $\F{\xf}\,\uf = \xf$ and $\ell(\xf,\uf)=0$. 
Since $\uellc{u} \geq0$ for all $u \in\mbbU^{\abs{\Omega}}$ and due to Assumption~\ref{assum:stage_cost_bounds}, there exists a function $\alpha_{1}=\alphaell{}\in\mcalK_{\infty}$ such that for all $x\in\feasibleset\subset\feasiblesetc$ and $u\in\mbbU^{\abs{\Omega}}$, we have that $ \alphaellnorm{x} \leq \xellc{x} \leq \ell(x,u)$. 
Since $\Jfc{}$ is continuous at the origin and locally bounded on the closed set $\terminalset=\levJf{\pi}$, there exists a function $\alpha_{\Jfc{}}\in\mcalK_{\infty}$ such that for all $x \in\terminalset$, we have that $\Jfc{x}\leq \alpha_{\Jfc{}}(\norm{x})$ \cite{techRawlings2017,Allan2017}.
Together with Assumptions~\ref{assum:constraint_sets} and \ref{assum:terminal_control_law}, we have all the common stabilizing conditions together to rely on the textbook proof \cite[Thm. 2.19 (a)]{Rawlings2020} to verify that $\VN{}$ is a valid Lyapunov function.
From a standard result \cite[Thm. 2.13]{Rawlings2020}, we deduce asymptotic stability.
\end{proof}

\subsection{Proof of Proposition~\ref{prop:practical_stability}}

The following derivations are mainly based on \cite{Allan2017,Rawlings2020}. To evaluate the difference between the evolutions of the nominal and the input perturbed systems in \eqref{eq:sys_cl} and \eqref{eq:sys_cl_d}, respectively, we rely on Lemma \ref{lem:lem_1}, wich is taken from \cite[Prop. 20]{Allan2017}.
\begin{lem}[See \cite{Allan2017}]
Defne $C  \subseteq D \subseteq X$ with $C$ compact and $D$ closed. 
Let $g\in\mcalC^{0}(D,X) $. 
Then there is a function $\alpha \in \mathcal{K}_{\infty}$ such that for all $x \in D$ and $y \in C$, we have that $\norm[in]{g(x)-g(y)} \leq \alpha(\norm[in]{x-y})$.
\label{lem:lem_1}
\end{lem}

Let us define the warm-start control sequence by
$\uswarm{x} = \brackets[eq]{\uopt{1}(x),\uopt{2}(x),\ldots,\uopt{N-1}(x), \xmuf{\xsol{N,x,\usopt{x}}}}$.
Let $\xdplus\coloneqq \xsolos{N_{\os},x,\omegasmu{x}}$ and $\xplus\coloneqq \F{x}\,\xmu{x}$.
\vspace{0.5\topsep}
\begin{proof}
The bounds $\alpha_{1}$ and $\alpha_{2}$ follow from Proposition~\ref{prop:nominal_stability} and rely on Assumptions~\ref{assum:steady-state}--\ref{assum:constraint_sets}. 
To proof $\mcalP$-practical asymptotic stability, we first combine the derivations in \cite{Allan2017} (robust asymptotic stability, suboptimal MPC) and \cite[Chap. 3.2.4]{Rawlings2020} (robust exponential cost decrease, optimal MPC). 
Then, we derive a new cost decay bound $\alpha_{3}\in\mcalK_{\infty}$. 

Let us consider any $\rho\geq0$ such that $\levVN{\rho}\subseteq \feasibleset$.
To rely on Lemma~\ref{lem:lem_1}, $\levVN{\rho}$ must be compact.
Since $J_{N}$ is lower-bounded by $\alpha_{\xellc{}} \in \mcalK_{\infty}$ for all $(x,\us) \in \feasibleset \times \mbbU^{\abs{\Omega}}$ where $\mbbU^{\abs{\Omega}}$ is compact, see Assumptions~\ref{assum:stage_cost_bounds}--\ref{assum:constraint_sets}, it is level-bounded in $x$ uniformly over $\us \in \mbbU^{\abs{\Omega}}\supseteq \admissiblecontrolset{x}$.  
Further, since $J_{N}$ is assumed to be continuous, $V_{N}$ is lower semicontinuous by \cite[Thm.\,1.17, Ex.\,1.19]{Rockafellar1998}.
We conclude that $\levVN{\rho}$ is closed and since $\norm{x}\leq \alpha_{\xellc{}}^{-1}(\rho)$ for all $x\in\levVN{\rho}$, $\levVN{\rho}$ is bounded and thus compact, see also \cite[Prop.\,9]{Kuntz2025}.

\subsubsection{\textbf{Terminal Constraint Satisfaction}}
We have continuous mappings $\xdplus\mapsto\xdplus[N]\coloneqq\xsol{N,\xdplus,\uswarm{x}}$ and $\xplus\mapsto\xplus[N]\coloneqq\xsol{N,\xplus,\uswarm{x}}$. Based on Lemma \ref{lem:lem_1}, for all $x, \xplus \in \levVN{\rho}$ and $\xdplus \in\feasibleset$, we have that (see \cite{Allan2017}):
\begin{equation}
\norm[eq]{\Jfc[eq]{\xdplus[N]} - \Jfc[eq]{\xplus[N]}}\leq \alpha_{\Jfc{}}\paren[eq]{\norm[eq]{\xdplus-\xplus}},\,\alpha_{\Jfc{}}\in\mcalK_{\infty}.
\end{equation}
Let us consider the upper bound for $\Jfc[eq]{\xdplus[N]}$ (see \cite{Allan2017}):
\begin{equation}
\Jfc[eq]{\xdplus[N]} \leq \Jfc[eq]{\xplus[N]}+ \alpha_{\Jfc{}}\paren[eq]{\norm[eq]{\xdplus-\xplus}}.
\end{equation}
After inserting the nominal cost decrease from Assumption~\ref{assum:terminal_control_law} with $x^{*}[N]\coloneqq\xsol{N,x,\usopt{x}}$, we obtain (see \cite{Allan2017}):
\begin{equation}
\Jfc[eq]{\xdplus[N]} \hspace{-0.1em}  \leq \Jfc[eq]{x^{*}[N]}-\alphaellnorm[eq]{x^{*}[N]}+\alpha_{\Jfc{}}\paren[eq]{ \norm[eq]{\xdplus - \xplus } }.
\end{equation}
If $\alpha_{\Jfc{}}\paren{\norm{x^{*}[N]}} \geq \Jfc{x^{*}[N]}\geq \pi/\tau$ with some $\tau \in \mbbR_{\geq 1}$, we have that $\norm{x^{*}[N]}\geq \alpha_{\Jfc{}}^{-1}\paren{\pi/\tau}$. The upper bound thus follows by (see \cite{Allan2017}):
\begin{equation}
\Jfc[eq]{\xdplus[N]} \hspace{-0.1em}  \leq \pi - \alphaell[eq]{\alpha_{\Jfc{}}^{-1}\paren[eq]{\frac{\pi}{\tau}} }+\alpha_{\Jfc{}}\paren[eq]{ \norm{\xdplus - \xplus } }.
\end{equation}
If $\Jfc[eq]{x^{*}[N]}< \pi/\tau$, we obtain the upper bound (see \cite{Allan2017}):
\begin{equation}
\begin{split}
\Jfc[eq]{\xdplus[N]} \leq~&\pi/\tau - 0+\alpha_{\Jfc{}}\paren[eq]{ \norm[eq]{\xdplus - \xplus } }. 
\end{split}
\end{equation}
If $\Jfc{\xdplus[N]} \leq \pi$, the warm-start is admissible with $\uswarm{x}\in\admissiblecontrolset{\xdplus}$. Therefore, $\xdplus\in\feasibleset$.
To robustly satisfy the terminal constraint, we therefore request that (cf. \cite{Allan2017}):
\begin{equation}
\begin{split}
&\norm[eq]{\xdplus - \xplus } \leq \gamma_{1}\coloneqq c_{1}(\pi)\coloneqq  \\ 
&\alpha_{\Jfc{}}^{-1}\paren[eq]{\pi -  \min_{\tau\in \mbbR_{\geq 1}} \max \cbraces[eq]{\frac{\pi}{\tau},\pi - \alphaell[eq]{\alpha_{\Jfc{}}^{-1}\paren[eq]{\frac{\pi}{\tau}}} }}.
\end{split}
\label{eq:e_1}
\end{equation}
Notice that if $\pi > 0$, we have that $\gamma_{1} > 0$, otherwise $\gamma_{1} \geq 0$. 
\subsubsection{\textbf{Positive Invariance}}
We again rely on Lemma~\ref{lem:lem_1} and examine the continuous finite horizon cost function for all $x, \xplus \in \levVN{\rho}$ and $\xdplus \in\feasibleset$ with $\alpha_{\JN{}}\in\mcalK_{\infty}$ (see \cite{Allan2017}):
\begin{equation}
\norm[eq]{ \JN[eq]{\xdplus,\uswarm{x}} -\JN[eq]{\xplus,\uswarm{x}}}\leq\alpha_{\JN{}}\paren[eq]{ \norm[eq]{\xdplus - \xplus } }.
\end{equation}

From Prop.~\ref{prop:nominal_stability}, we deduce $\VN{\xplus}\leq \JN{\xplus,\uswarm{x}} \leq \VN{x} - \alphaellnorm{x}$ and obtain (cf. \cite[Chap. 3.2.4]{Rawlings2020}):
\begin{equation}
\VN{\xdplus} \leq \VN{x} - \alphaellnorm{x} + \alpha_{\JN{}}\paren{ \norm{\xdplus - \xplus } }.
\label{eq:cost_decay}
\end{equation}
According to \cite{Allan2017}, if $\alphatwonorm{x}\geq \VN{x}\geq \rho/\tau$, we have that $\norm{x}\geq \alphatwoinv{\rho/\tau}$, such that the upper bound is:
\begin{equation}
\VN{\xdplus} \leq \rho - \alphaell[eq]{\alphatwoinv[eq]{\frac{\rho}{\tau}}}+ \alpha_{\JN{}}\paren[eq]{\norm[eq]{\xdplus - \xplus } }.
\end{equation}
If $\VN{x}< \rho/\tau$, $\VN{\xdplus}$ is upper bounded by (see \cite{Allan2017}):
\begin{equation}
\VN{\xdplus} \leq \rho/\tau - 0 + \alpha_{\JN{}}\paren{ \norm{\xdplus - \xplus } }.
\end{equation}
If $\VN{\xdplus} \leq \rho$, $\levVN{\rho}$ is positive invariant for system~$\eqref{eq:sys_cl_d}$ with  $\xsolos{N_{\os},x,\omegasmu{x}}\in \levVN{\rho}$ for all $x \in \levVN{\rho}$. To ensure positive invariance, we thus require that (cf. \cite{Allan2017}):
\begin{align}
&\norm[eq]{\xdplus - \xplus }  \leq \gamma_{2}\coloneqq c_{2}(\rho) \coloneq \\
&\alpha_{\JN{}}^{-1}\paren[eq]{\rho -  \min_{\tau\in \mbbR_{\geq 1}} \max \cbraces[eq]{\frac{\rho}{\tau},\rho - \alphaell[eq]{\alphatwoinv[eq]{\frac{\rho}{\tau}}}}}.
\nonumber
\label{eq:e_2}
\end{align}
Again, if $\rho>0$, we have that $\gamma_{2} >0$.

For every $\kappa \in (0,\rho)$, $\levVN{\kappa}$  is a positive invariant set for system $\eqref{eq:sys_cl_d}$ if
$\norm[in]{\xdplus - \xplus }  \leq \gamma_{3}\coloneqq c_{2}(\kappa)$ (cf. \cite[Sec. 3.2.4]{Rawlings2020}).

\subsubsection{\textbf{Cost Decay}}
For a cost decay in the asymptotic sense, we want to ensure that $\VN{\xdplus}<\VN{x}$ for $\norm{x}>0$. From~\eqref{eq:cost_decay}, we deduce the requirement (cf. \cite[Chap. 3.2.4]{Rawlings2020}):
\begin{equation}
\norm[eq]{\xdplus - \xplus }  < \alpha_{\JN{}}^{-1}\paren{\alphaellnorm{x} }.
\end{equation}
Notice that for $\norm{x}\to 0$, the perturbation bound also tends to zero. At the boundary of $\levVN{\kappa}$, we have that $\VN{x}=\kappa \leq \alphatwonorm{x}$. 
From this, we obtain $\norm{x} \geq \delta\coloneqq\alphatwoinv{\kappa}$ and the more conservative requirement (cf. \cite[Chap. 3.2.4]{Rawlings2020}):
\begin{equation}
\norm[eq]{\xdplus - \xplus }  \leq \gamma_{4} \coloneqq d(\kappa)\coloneqq\alpha_{\JN{}}^{-1}\paren{\alphaell{\alphatwoinv{\kappa}}}.
\label{eq:cost_decay_bound}
\end{equation}
If $\kappa > 0$, we have that $\gamma_{4}>0$.
\subsubsection{\textbf{Practical Asymptotic Stability}}
To define a proper cost decay of the optimal value function, we must consider the worst-case scenario $\gamma_{4} = \min \cbraces{\gamma_{1},\gamma_{2},\gamma_{3},\gamma_{4}}$. 
Recall that $\alphaell{s} - \alpha_{\JN{}}\paren{\gamma_{4}}> 0$ only holds as long as $s>\delta$, see \eqref{eq:cost_decay} and \eqref{eq:cost_decay_bound}. 
Therefore, we use an infinitesimally small part of the robustness margin $0<\varepsilon \ll \min \cbraces{\gamma_{1},\gamma_{2},\gamma_{3},\gamma_{4}} $ to lift $\alphaell{\delta} - \alpha_{\JN{}}\paren{\gamma_{4}}$ by introducing \mbox{$\gamma \coloneqq \min \cbraces{\gamma_{1},\gamma_{2},\gamma_{3},\gamma_{4}} -\varepsilon>0$}. 
Let us define the cost decay function
\begin{equation}
\alpha_{3}(s)\coloneqq\begin{cases}
\alpha(s) & \mrm{if} ~ s< \delta,\\
\alphaell{s} - \alpha_{\JN{}}\paren{\gamma} &\mrm{if} ~ s\geq \delta,\\
\end{cases}
\end{equation}
where $\alpha \in \mcalK$ shall satisfy $\alpha(\delta)= \alphaell{\delta} - \alpha_{\JN{}}\paren{\gamma}$. 
The specific choice of $\alpha$ is not relevant since we only strive for $\mcalP$-practical asymptotic stability. 
We claim that $\alpha_{3}\in\mcalK_{\infty}$. 

We summarize that the optimal value function $\VN{}$ is a valid Lyapunov function on the set $ \levVN{\rho}\backslash\mathrm{int}(\mcalB_{\delta})$ for system~\eqref{eq:sys_cl_d} with $\mcalB_{\delta}\subseteq\levVN{\kappa}$, where $\mathrm{int}(\mcalB_{\delta})$ is the interior of $\mcalB_{\delta}$. 
The closed ball $\mcalB_{\delta}$ is a subset of $\levVN{\kappa}$ since $\norm{x} = \delta=\alphatwoinv{\kappa}$ in \eqref{eq:cost_decay_bound} is the minimum distance from the origin to the boundary of $\levVN{\kappa}$. 
The control law $\omegasmu{}$ renders the sets $ \levVN{\rho}$ and $\mcalP=\levVN{\kappa}$ positive invariant for system~\eqref{eq:sys_cl_d} if it adheres to the perturbations bounds $\min\{\gamma_{1},\gamma_{2}\}$ and $ \min\{\gamma_{1},\gamma_{3}\}$, respectively.
If we choose $\gamma \coloneqq \min \cbraces{\gamma_{1},\gamma_{2},\gamma_{3},\gamma_{4}} -\varepsilon>0$,
\mbox{$\mcalP$-practical} asymptotic stability follows from the well-known result \cite[Thm. 2.20]{Gruene2017}.
\end{proof}

\end{appendix}

\bibliographystyle{IEEEtran}       
\bibliography{literature.bib}    

\end{document}